\documentclass[letterpaper, 10 pt, conference]{ieeeconf}  

\IEEEoverridecommandlockouts                              
\usepackage{amsmath} 
\usepackage{amssymb}  
\usepackage{algorithm}
\usepackage{algpseudocode}
\usepackage{threeparttable}

\newtheorem{lemma}{Lemma}
\newtheorem{proposition}{Proposition}

\usepackage{booktabs}
\usepackage{multirow}
\usepackage{mwe}
\usepackage{svg}

\usepackage{acronym}
\usepackage{tikz,pgfplots} 
\pgfplotsset{compat=1.18}
\usepackage{macros}

\title{\LARGE \bf
Experimental Design for Controller Selection in Synthetic Biology
}

\author{Eric Palanques-Tost$^1$, Ron Weiss$^2$, Calin Belta$^3$
\thanks{$^1$E. Palanques-Tost is with Boston University, Boston, MA, USA {\tt\small ericpt@bu.edu}}
\thanks{$^2$R. Weiss is with the Massachusetts Institute of Technology, Cambridge, MA, USA {\tt\small rweiss@mit.edu}}
\thanks{$^3$C. Belta is with the University of Maryland, College Park, MD, USA {\tt\small calin@umd.edu}}
}

\begin{document}

\bstctlcite{BSTcontrol}

\maketitle
\thispagestyle{empty}
\pagestyle{empty}


\begin{abstract} 
Synthetic biology enables the design of genetic circuits that act as feedback controllers. These controllers are typically designed using computational models, but mismatch between model and real dynamics can lead to controllers that fail in practice. While methods to address this issue exist, synthetic biology introduces additional structural constraints. Genetic circuits are often highly constrained by experimental limitations, reducing controller design to selection among a limited set of implementable circuits rather than an optimization over a continuous space. As a result, multiple system hypotheses may lead to the same optimal controller within the implementable set. Reducing model uncertainty may therefore be irrelevant when the models lead to the same optimal controller. In this paper, we exploit this structure to develop an algorithm for controller selection in synthetic biology, formulating the problem as a decision-oriented experimental design problem over a finite controller set. We represent plant uncertainty using a set of hypotheses and select experiments to minimize the posterior controller selection risk, rather than global model uncertainty. Across three mechanistic case studies, our method reaches the stopping criterion in fewer experimental rounds than model uncertainty and random experiment selection policies, while maintaining a comparable success rate.
\end{abstract}

\section{Introduction}

Living cells continuously sense molecules in their environment and produce molecular signals in response,  influencing other cells around them. Through these interactions, cells regulate biological systems and their dynamical responses to the environment. This behavior is analogous to how electronic feedback controllers interact with dynamical systems in closed loop to stabilize them or guide them toward a desired behavior. In recent years, synthetic biology has made it possible to modify cellular regulatory interactions through genome editing, allowing researchers to engineer control functions in genetic networks. Examples include population control circuits that stabilize cell density~\cite{youProgrammedPopulationControl2004}, metabolic control circuits that regulate production of desired metabolites~\cite{guptaDynamicRegulationMetabolic2017}, and biomolecular circuits implementing integral feedback to achieve robust homeostasis (i.e., maintain a regulated variable at a desired level despite disturbances)~\cite{briat2016-antitheticFeedback}. In these systems, the implemented genetic circuit senses molecular signals from the cell and produces regulatory molecules that modify cellular behavior, forming a closed-loop dynamical system.

Designing such biological controllers is typically done \textit{in silico}, using computational models of the biological system to predict the behavior of candidate controllers before experimental  implementation~\cite{palanques-tostSTLbasedOptimizationBiomolecular2025}. However, the dynamics of real biological systems are only partially understood, and models often simplify interaction mechanisms or neglect relevant molecular species. As a result, controllers optimized in simulation may fail when deployed experimentally. This motivates methods that account for plant uncertainty in controller design.

Several frameworks address uncertainty in control. Robust control constructs controllers that perform acceptably across a range of possible models, often at the cost of conservatism~\cite{gutmanRobustAdaptiveControl2003, petersenRobustControlUncertain2014}. Bayesian Experimental Design (BED) instead identifies  informative experiments that reduce system uncertainty~\cite{fosterUnifiedStochasticGradient2020}. 
Recent goal-oriented or decision-focused approaches select experiments to reduce uncertainty relevant to a downstream task, rather than improving model accuracy globally~\cite{kandasamyMyopicPosteriorSampling2019, huangAmortizedBayesianExperimental2024, zhongGoalOrientedBayesianOptimal2026}. 

In synthetic biology, however, controller implementation introduces additional constraints.
First, controllers often cannot be freely parameterized on a continuous scale, but must be constructed from a limited set of biological components, resulting in a finite library of implementable genetic circuits. Second, control objectives are often specified in terms of acceptable qualitative regimes or concentration ranges, rather than precise setpoints. As a consequence, different models of the system may lead to the same optimal controller within the feasible library. Reducing model uncertainty may therefore have little impact on controller selection when the competing models agree on the same optimal controller. 

In this paper, we formulate biological controller selection as a specialization of goal-oriented experimental design for closed-loop dynamical systems. We represent plant uncertainty using a finite set of candidate hypotheses and restrict controller design to a finite library of implementable genetic circuits. We maintain a posterior over plant hypotheses and select experiments based on their expected reduction in controller regret. Therefore, rather than identifying the true system dynamics, experiments are selected to resolve ambiguity in controller selection.

Specifically, the contributions of this paper are:
\begin{itemize}
	\item We formulate controller selection in synthetic biology as a goal-oriented experimental design problem over a finite library of implementable controllers. 
    \item We develop an algorithm for controller selection that directs experiments toward uncertainties affecting controller selection.
    \item We evaluate the method in three mechanistic case studies against model-uncertainty and random experiment-selection policies.
\end{itemize}

\section{Related Work} \label{sec:related-work}
Designing controllers that perform well under model uncertainty is a central problem in control. Robust and adaptive control address model uncertainty by synthesizing controllers that maintain acceptable performance across plausible system behaviors or disturbances~\cite{gutmanRobustAdaptiveControl2003, petersenRobustControlUncertain2014}. Identification-for-control emphasizes that system identification should be guided by the downstream control objective rather than model accuracy alone~\cite{geversIdentificationControlEarly2005}. 

Bayesian Experimental Design (BED) provides a general framework for selecting informative experiments through data acquisition~\cite{fosterUnifiedStochasticGradient2020, liepeMaximizingInformationContent2013}. Classical BED typically targets uncertainty about model parameters or latent variables, whereas recent goal-oriented or decision-focused extensions optimize experiments according to their impact on a downstream decision~\cite{kandasamyMyopicPosteriorSampling2019, huangAmortizedBayesianExperimental2024, zhongGoalOrientedBayesianOptimal2026}. In dynamical systems, control-oriented experiment design specializes this approach to select experiments according to their effect on closed-loop performance~\cite{andersonControlOrientedIdentificationStochastic2024, ebadatModelPredictiveControl2017}. Our work builds on this decision-oriented perspective, but focuses on settings inspired by synthetic biology.

Experimental design has already been used in systems biology for parameter estimation and model discrimination~\cite{liepeMaximizingInformationContent2013, flassigOptimalDesignStimulus2012, busettoNearoptimalExperimentalDesign2013a}, and in synthetic biology to distinguish mechanistic models or improve gene-regulatory-network predictions~\cite{bandieraOptimallyDesignedModel2020, braniffOptimalExperimentalDesign2019}. These works primarily target model or parameter identification. Here, multiple plausible models may induce the same optimal controller, so we instead design experiments around uncertainty that affects controller selection.

\section{Biological control}\label{sec:biological-control}

In this paper, we study biological control systems where the objective is to regulate the concentration of biological species (e.g. proteins) over time. Control is implemented through the introduction of additional biological species that interact and modulate the system.

Formally, let $x \in \mathbb{R}^n_{\ge 0}$ denote the concentration of species of a biological system (the \emph{plant}), whose dynamics are described by a function $f$. Control in this system can be introduced through engineered species with concentrations $z\in \mathbb{R}^m_{\ge 0}$, which interact with $x$ and whose dynamics are governed by a function $g$. The combined system dynamics can be expressed as follows: 

\begin{equation} \label{eq:bio-dynamics}
    \begin{aligned}
    \dot{x}(t) &= f(x(t), z(t), d(t)), \\
    \dot{z}(t) &= g(x(t), z(t)),
    \end{aligned}
\end{equation}
where $f:\mathbb{R}^n\times\mathbb{R}^m\times\mathbb{R}^l\rightarrow\mathbb{R}^n$ captures the dynamics of the fixed biological system and $g:\mathbb{R}^n\times\mathbb{R}^m\rightarrow\mathbb{R}^m$ defines the controller dynamics. $d(t) \in \mathbb{R}^l$ is a perturbation signal representing environmental effects on the plant.

Unlike classical formulations where the control input is externally imposed, here control is realized through the choice of controller dynamics $g$, whose species $z$ interact with $x$. We refer to systems of this form as \emph{biological feedback control systems}.

\textbf{Example} Consider a biological system with species $x=[A,I]$, where $I$ is a harmful damage marker, and $A$ is a therapeutic molecule that reduces $I$. This system, with dynamics described by $f$, represents the plant. A biological controller may be an engineered genetic circuit with species $z$ and dynamics $g$. The circuit can be engineered to produce $A$ as a response to $I$,  forming the closed-loop dynamics in~\eqref{eq:bio-dynamics}.

\section{Problem formulation} \label{sec:problem-formulation}

\subsection{Constraint-free objective}
We consider the problem of selecting a biological controller for a dynamical system as described in Sec.~\ref{sec:biological-control}. Let $\mathcal{C}$ denote a finite set of candidate biological controllers, where each $c\in \mathcal{C}$ induces fixed controller dynamics $g_c$.

Let $J(f,g)$ denote a function that evaluates the closed-loop trajectories obtained by integrating Eq.~(\ref{eq:bio-dynamics}), according to a specified control objective (e.g. regulate the quantity of a protein) where higher $J$ means better performance. Our goal is to identify the controller $c\in \mathcal{C}$ that maximizes $J(f,g_c)$. Formally,

\begin{equation} \label{eq:control}
c^\star = \arg\max_{c \in \mathcal{C}} J(f,g_c).
\end{equation}

\textbf{Example (continued)} 
Consider a system with two species $x=[A,I]$ governed by dynamics $f$. We want to find a controller $c\in\mathcal{C}$ with dynamics $g_c$ that minimizes the concentration $I$ by producing $A$, which we express as maximizing the function $J(f,g_c)=-\int_0^T I(t)dt$ penalizing the cumulative concentration of $I$. $\mathcal{C}$ represents the set of implementable genetic circuits.

\subsection{Constrained objective} \label{sec:constrained-objective}

We consider the setting where the system dynamics $f$ are unknown. Instead, we are given a finite set of candidate models $\mathcal{H}$, each representing a possible hypothesis for the dynamics. Assume that there exists a hypothesis $h^\star \in \mathcal{H}$ representing the true dynamics such that $f:=h^\star$.

Information about the system is obtained through experiments. Let $\mathcal{E}$ denote the set of feasible experiments, where each experiment $e \in \mathcal{E}$ is defined by a tuple $e = (x_0, d)$ consisting of (i) an initial condition $x_0$, and (ii) an externally imposed perturbation $d(t)$, such as experimentally adding one of the species at a prescribed rate. Given a system $f$ and a controller $c$, executing $e$ produces a set of observations $y =\{y(\tau_i)\}_{i=1}^{T_O}$ from the resulting closed-loop system formed by $f$ and $g_c$, corresponding to (possibly noisy and partial) measurements of selected variables at discrete time points. 

Experiments are performed sequentially. At each round $k$, an experiment $e_k \in \mathcal{E}$ is selected and executed. Based on the observations collected up to round $k$, we update our belief $\pi_k$ over the candidate models, representing our current uncertainty about which model $h \in \mathcal{H}$ describes the system. We then select a controller that maximizes the expected performance based on the current belief $\pi_k$:
\begin{equation}
\hat{c}_k = \arg\max_{c \in \mathcal{C}} \mathbb{E}_{h \sim \pi_k}[J(h, g_c)].
\end{equation}

Let $\rho$ denote the round at which the procedure stops. Our objective is:
\begin{equation}
    \min \rho \quad \text{s.t.} \quad \hat c_\rho = c^\star,
\end{equation}
where $c^\star$ is the optimal controller for the true system obtained as in Eq.~\ref{eq:control}.

\textbf{Example (continued)} Continuing the example above, suppose there are two candidate models for the system dynamics $\mathcal{H} = \{h_0, h_1\}$, where $h_0$ assumes that $A$ is always beneficial, while $h_1$ captures a regime in which $A$ becomes harmful at high $I$. The controller set $\mathcal{C}$ contains $c_0$, which produces $A$ proportionally to $I$, and $c_1$, which produces $A$ only within a specified range of $I$. The experimental library $\mathcal{E}$ is defined by combining two initial conditions (low and high $I$), and three external disturbances externally adding $I$ at different rates, resulting in $6$ possible experiments ($x_0$, $d$). Our objective is to identify a controller $c \in \mathcal{C}$ that minimizes $I$ under the true system dynamics using as few experiments as possible.
\section{Methods} \label{sec:methods}

\begin{figure*}
    \centering
    \includegraphics[width=\linewidth]{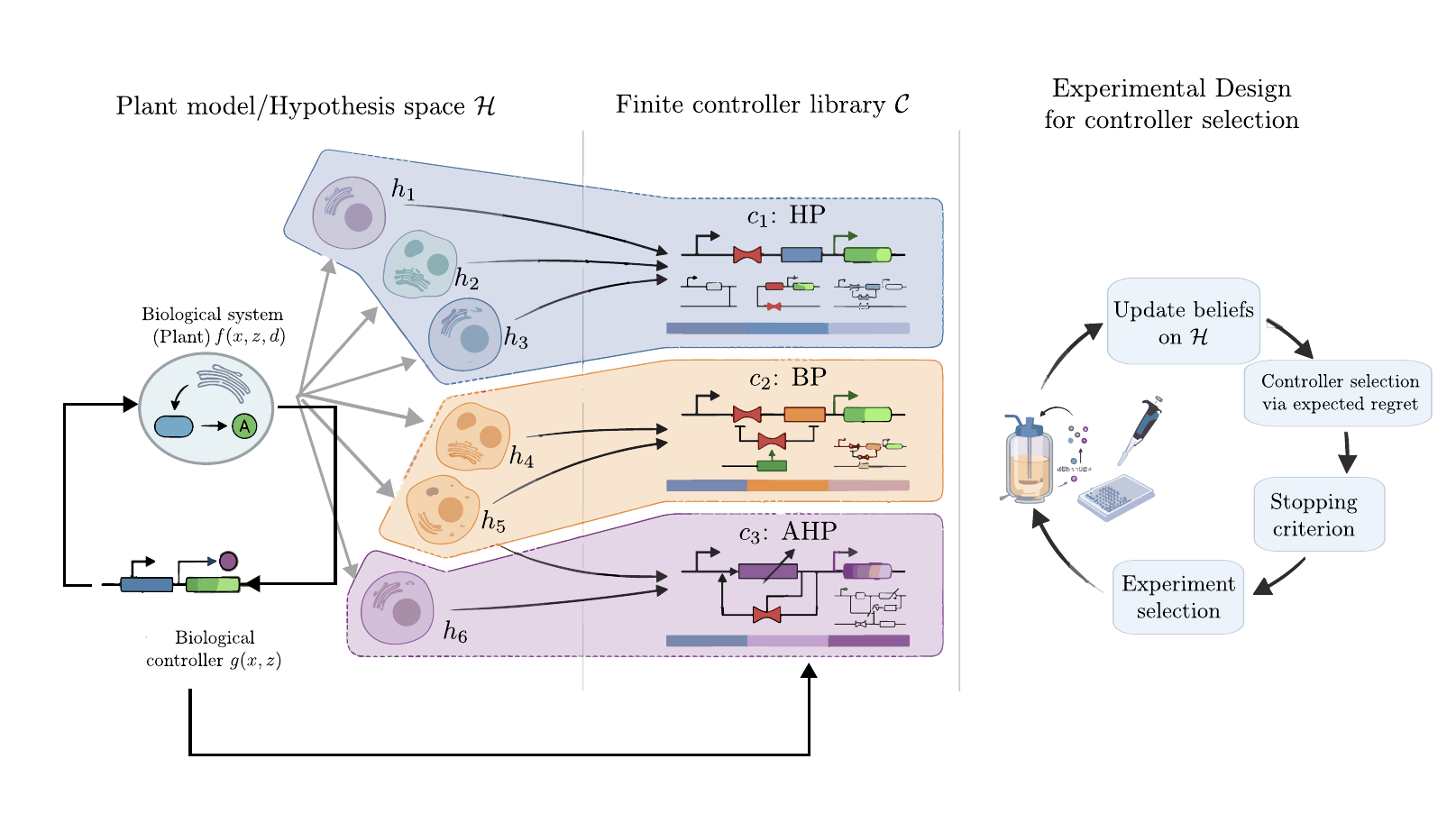}
    \caption{\textbf{Overview of the proposed framework}. \textbf{Left}: A biological system with unknown dynamics $f$ is described by a finite set of candidate hypotheses $\mathcal{H}$. \textbf{Middle}: A finite library $ \mathcal{C}$ defines the implementable genetic controllers. Multiple hypotheses may share the same optimal controller. \textbf{Right}: Experiments are selected to reduce posterior expected controller regret, stopping once remaining model uncertainty has limited impact on controller performance.}
    \label{fig:overview}
\end{figure*}

We propose a framework to address the problem stated in Sec.~\ref{sec:constrained-objective} that alternates between (i) inference over plausible system dynamics, (ii) controller selection under uncertainty, and (iii) experiment design to reduce decision-relevant uncertainty, as shown in Fig.~\ref{fig:overview} and Alg.~\ref{alg:controller-selection}.

\subsection{Belief over the hypothesis set}
At each experimental round $k$, we update a belief $\pi_k(h)$ over the set of candidate models $h\in\mathcal{H}$ for the true dynamical system. Each element $h \in \mathcal{H}$ defines a dynamical function that can be evaluated in simulation.

Let $\mathcal{D}_k = \{(y_j, c_j, e_j)\}_{j=1}^k$ denote the dataset of observations collected up to round $k$. Each observation $y_j = \{y_j(\tau_i)\}_{i=1}^{T_O}$ corresponds to measurements obtained by applying controller $c_j$ under experiment $e_j$, at discrete time points $\{\tau_i\}_{i=1}^{T_O}$.

We assume that observations from different experimental rounds are conditionally independent given the hypothesis $h$ and the applied controller–experiment pairs $(c_j,e_j)$. The posterior over models is then updated using Bayes' rule:
\begin{equation} \label{eq:model-posterior}
    \pi_k(h) \propto p(\mathcal{D}_k \mid h)\, \pi_0(h),
\end{equation}
where $p(\mathcal{D}_k \mid h) = \prod_{j=1}^k p(y_j \mid h, c_j, e_j)$. We assume the likelihood $p(y_j\mid h,c_j,e_j)$ can be evaluated for each hypothesis, controller, and experiment.

\subsection{Controller selection via expected regret}

Under model uncertainty, selecting a controller based on a single model can lead to suboptimal performance. Instead, we evaluate controllers based on their expected regret under the posterior $\pi_k(h)$.

For a given model $h \in \mathcal{H}$, we define the model's optimal controller as $c^{\mathrm{opt}}(h) = \argmax_{c' \in \mathcal{C}} J(h, g_{c'})$. Then, we define the regret of choosing an alternative controller $c \in \mathcal{C}$ as the difference in performance between the optimal and the alternative controller:
\begin{equation} \label{eq:regret}
    r(c, h) = J(h, g_{c^{\mathrm{opt}}(h)}) - J(h, g_c),
\end{equation}

The posterior expected regret of controller $c$ at round $k$ is then:
\begin{equation} \label{eq:expected-regret}
    R_k(c) = \mathbb{E}_{h \sim \pi_k} [r(c, h)] = \sum_{h \in \mathcal{H}} \pi_k(h)\, r(c, h).
\end{equation}

Then, the controller selected at round $k$ is the controller that minimizes the expected regret:
\begin{equation} \label{eq:argmin-regret}
    \hat{c}_k = \argmin_{c \in \mathcal{C}} R_k(c)
\end{equation}

This criterion accounts not only for which controller is optimal under each model, but also for the performance degradation incurred when selecting a suboptimal controller.

\subsection{Stopping criterion}

Since the true optimal controller $c^\star$ is unknown, the constraint $\hat c_\rho = c^\star$ cannot be verified directly. Instead, we terminate the algorithm at round $k$ when the minimum expected regret is below a predefined threshold $\varepsilon$, ensuring that any remaining model uncertainty has a limited impact on the chosen controller.

\begin{equation} \label{eq:stopping}
    \min_{c \in \mathcal{C}} R_k(c) \le \varepsilon.
\end{equation}

\subsection{Experiment selection via regret reduction}

When the stopping criterion is not met, we select the experiment with the largest expected reduction in minimum regret. Let $\hat{c}_k=\arg\min_cR_k(c)$ denote the current controller. For each $e\in\mathcal E$, we approximate its expected outcome via Monte Carlo sampling.

Specifically, we sample outcomes $y$ from the predictive distribution $p(y \mid e, \hat{c}_k, \pi_k) = \sum_{h\in\mathcal H} \pi_k(h) p(y\mid h,\hat{c}_k,e)$ by sampling a model $h\sim\pi_k$ and then obtaining $y\sim p(y\mid h,\hat{c}_k,e)$. For each sampled outcome $y$, we update the posterior $\pi_{k+1}(h \mid y,\hat{c}_k,e)$ as in Eq.~\ref{eq:model-posterior} and recompute the expected regret $R_{k+1}(c\mid y, \hat{c}_k, e)$ as in
Eq.~\ref{eq:expected-regret}. The utility of an experiment is then defined as:
\begin{equation} \label{eq:exp-utility}
    U_k(e) = \min_{c \in \mathcal{C}} R_k(c)
    -
    \mathbb{E}_{y \sim p(y \mid e, \hat{c}_k, \pi_k)}
    \left[
        \min_{c \in \mathcal{C}} R_{k+1}(c \mid y,\hat{c}_k, e)
    \right]
\end{equation}

The next experiment is then chosen as $e^\star_k = \argmax_{e \in \mathcal{E}} U_k(e)$, which prioritizes experiments that are most expected to reduce the risk of selecting a suboptimal controller.

\begin{algorithm}[t]
\caption{Experimental design for controller selection}
\label{alg:controller-selection}
\begin{algorithmic}[1]
\Require Hypotheses $\mathcal H$, prior $\pi_0$, controller set $\mathcal C$,
experiment set $\mathcal E$, threshold $\epsilon$

\State Initialize dataset $\mathcal D \gets \emptyset$
\State For each $h\in\mathcal H$: 
$c^{\mathrm{opt}}(h)\gets\argmax_{c\in\mathcal C}J(h,g_c)$

\While{true}
    \State Update posterior $\pi(h)$ from $\mathcal D$ using Eq.~\eqref{eq:model-posterior}
    \State Get expected regret $R(c)$ for all $c\in\mathcal C$
    with Eq.~\eqref{eq:expected-regret}
    \State Select controller
    $\hat c\gets\arg\min_{c\in\mathcal C}R(c)$

    \If{$R(\hat c)\leq\epsilon$}
        \State \Return $\hat c$
    \EndIf

    \State Compute $U(e)$ for all $e\in\mathcal E$
    with Eq.~\eqref{eq:exp-utility}
    \State Select experiment
    $e^\star\gets\arg\max_{e\in\mathcal E}U(e)$
    \State Execute $(e^\star,\hat c)$ and observe $y$
    \State Update
    $\mathcal D\gets\mathcal D\cup\{(y,\hat c,e^\star)\}$
\EndWhile

\end{algorithmic}
\end{algorithm}

\subsection{Conditional stopping guarantee}

Let us extend Eqs.~\eqref{eq:expected-regret}~and~\eqref{eq:argmin-regret} to an arbitrary posterior $\pi$: 

\begin{equation} \label{eq:regret-by-posterior}
    R^\star(\pi)= \min_{c\in \mathcal{C}}\sum_{h\in \mathcal{H}}\pi(h)r(c,h)
\end{equation}

so that $R^\star(\pi_k)\equiv R_k(\hat{c}_k)$. Similarly, let $U_\pi(e)$ denote the utility in Eq.~\eqref{eq:exp-utility} with posterior $\pi$, so that $U_{\pi_k}(e) \equiv U_k(e)$.

\begin{lemma}
    Assume finite $\mathcal H,\mathcal C$, and finite-valued $J$. Then $R^\star(\pi)$ is concave in $\pi$. 
\end{lemma}
\begin{proof}
   $R^\star$ is the pointwise minimum of $\sum_h \pi(h)r(c,h)$, which is linear in $\pi$ for each $c$, and therefore it is concave. 
\end{proof}

\begin{proposition}
    Under exact Bayesian updating, $U_\pi(e)\geq 0$ for every $e\in\mathcal E$.
\end{proposition}

\begin{proof}
    Let $\pi'$ denote the posterior after experiment $e$. Since $\mathbb E[\pi'\mid\pi,e]=\pi$, Jensen's inequality and concavity of $R^\star$ give $\mathbb E[R^\star(\pi')\mid\pi,e]\le R^\star(\pi)$, and hence \(U_\pi(e)\ge0\).
\end{proof}

Let us define $ \delta_\epsilon=
\inf_{\pi:R^\star(\pi) > \epsilon}\max_{e\in\mathcal E}U_\pi(e)$, as the worst-case expected improvement of the best experiment before the stopping condition. 
\\
\begin{proposition}
    Under exact evaluation of $U_\pi(e)$, if $\delta_\epsilon>0$, the stopping round~$\tau_{\epsilon}=\inf\{k\geq 0:R^\star(\pi_k)\leq\epsilon\}$~reached by greedy acquisition with repeatable experiments satisfies $\mathbb E[\tau_\epsilon]\leq \frac{R^\star(\pi_0)}{\delta_\epsilon}$.   
\end{proposition}

\begin{proof}
    For every $k<\tau_\epsilon$, $R^\star(\pi_k)>\epsilon$, so by definition of $\delta_\epsilon$ and greedy selection, the expected regret decreases by at least $\delta_\epsilon$ per non-terminal round: $U_{\pi_k}(e_k^\star)\geq\delta_\epsilon$. Summing these decreases up to round $\tau_\epsilon$ gives $\delta_\epsilon \mathbb E[n\wedge\tau_\epsilon] \leq R^\star(\pi_0)$. With $n\to\infty$, $\mathbb E[\tau_\epsilon] \leq R^\star(\pi_0) / \delta_\epsilon$. 
\end{proof}

\section{Experiments}

We evaluate whether the method proposed in Sec.~\ref{sec:methods} can identify the optimal controller from a finite library using few experimental rounds. 

\subsection{Benchmark cases}

We consider three forms of model mismatch: amplitude-dependent effects, context-dependent interactions, and temporal adaptation. Each $H_i$ is a mechanistic family containing multiple models with different parameterizations, with $\mathcal H=\bigcup_i H_i$.

Across benchmarks, the objective is to reduce the damage marker $I$ while limiting therapeutic effort $A$, expressed as $J=-\int (I(t)^2 + A(t)^2) dt $. The controller library contains high-pass (HP), band-pass (BP), adaptive high-pass (AHP) which attenuates output under sustained activation, and multi-input variants of these motifs. For each benchmark, the initial condition $x_0$ is fixed, and experiments differ only in the externally applied disturbance $d(t)$, which enters the dynamics as an additive input to the damage species $I$. In Case 1, $\mathcal{E}$ contains seven step disturbances with different amplitudes and onset times. Case 2 contains ten step disturbances. Case 3 contains eight step and pulse disturbances with different onset times. Each experiment $e=(x_0,d)$ starts from $x_0$, applies an external perturbation $d(t)$ and observes $I$ over $T=20$ at $t=0,10,20$. Observations are corrupted with Gaussian noise $\eta\sim\mathcal N(0,0.09^2)$. We set the stopping threshold to $\epsilon=0.05$.

We compare our acquisition policy with model-uncertainty and random baselines. All methods share the posterior update, controller-selection rule, and stopping criterion. The model-uncertainty baseline maximizes expected reduction in hypothesis entropy $H(\pi)=-\sum_h\pi(h)\log\pi(h)$, while random selection samples uniformly from $\mathcal E$.

\paragraph{Case 1: Amplitude-dependent mismatch}

This case captures settings where the therapeutic $A$ becomes harmful when damage $I$ is very high. Biologically, this can occur in inflammatory systems, where an anti-inflammatory response may interfere with endogenous protective mechanisms during severe inflammation. $H_0$ assumes treatment is always beneficial, favoring HP control. $H_1$ assumes treatment remains beneficial but with reduced benefit under sustained activation, so the treatment penalty favors AHP control. $H_2$ assumes treatment becomes harmful at high $I$, which corresponds to the true plant, and favors BP control. The hypotheses agree at low $I$ and diverge at high amplitudes, so informative experiments must probe this regime.

\paragraph{Case 2: Context-dependent response}

This case represents systems where treatment response depends on an unobserved biological species $B$. The controller does not directly sense the damage marker $I$, but instead observes a proxy $P$. At low $B$, the therapeutic molecule $A$ suppresses $P$ and $I$. At high $B$, however, $A$ becomes harmful, increasing $I$ in a way that is not reflected in $P$. Thus, for the same observed $P$, treatment with $A$ may be beneficial or harmful depending on $B$. $H_0$ assumes no effect from $B$, favoring HP control. $H_1$ assumes that $B$ increases damage (both $P$ and $I$), so stronger HP control remains effective. $H_2$ correctly assumes that $B$ affects the efficacy of $A$, but that $B$ can be inferred indirectly through $P$, favoring BP control on $P$ alone. $H_3$ assumes that the effect of $B$ cannot be inferred from $P$ alone, favoring a two-input controller with HP control on $P$ and LP control on $B$. The true plant belongs to $H_3$. The hypotheses differ only when contextual effects are active, so informative experiments must reveal the influence of $B$.

\paragraph{Case 3: Dynamical adaptation mismatch}
This case captures adaptive biological responses in which sustained treatment induces tolerance, reducing treatment effectiveness over time. Such behavior is common in biology, for example in drug resistance. $H_0$ assumes no buildup of resistance, favoring HP control. $H_1$ includes tolerance dynamics with different strengths and persistence, favoring AHP control. The true plant belongs to $H_1$. The hypotheses differ mainly in their temporal response rather than steady-state behavior.

\subsection{Results}

\begin{figure}
    \centering
    \includegraphics[width=1\linewidth]{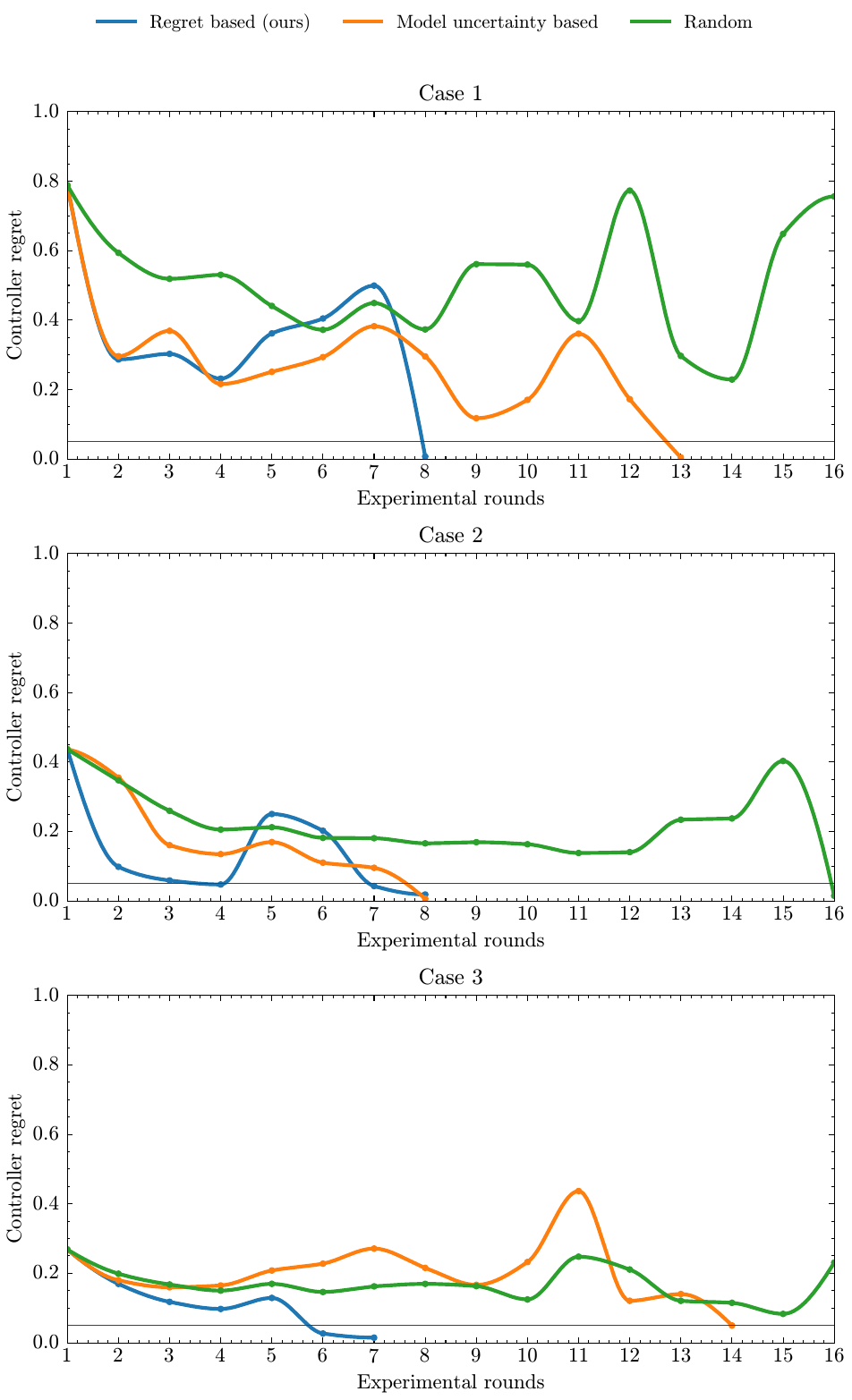}
    \caption{Controller regret across experimental rounds for the three benchmark cases. Curves show the mean over 50 random seeds. Methods differ only in the criteria for experiment selection (blue: minimize controller regret, ours; orange: minimize model uncertainty, green: random)}
    \label{fig:results-figure}
\end{figure}

Fig.~\ref{fig:results-figure} shows the evolution of controller regret across rounds in all three cases averaged over $50$ random seeds. Table~\ref{tab:results-table} reports the average number of rounds required to reach the stopping criterion, with the corresponding success rates. Success requires both reaching a controller regret below $0.05$ and selecting the correct optimal controller. 

The proposed method requires fewer rounds than the model-uncertainty baseline in all three cases while achieving similar success rates. For each comparison, runs generated from the same seed were paired when both methods succeeded. Two-sided paired \(t\)-tests were applied to the number of rounds, with Holm correction across all nine comparisons. Our method required significantly fewer rounds than the model-uncertainty baseline in all three cases, with a mean reduction of $0.57$ (95\% CI: $[0.07,1.08]$, $p_{\mathrm{adj}}=0.042$) in Case 1, $1.39$ ($[0.92,1.85]$, $p_{\mathrm{adj}}<0.001$) in Case 2, and $0.90$ ($[0.19,1.60]$, $p_{\mathrm{adj}}=0.042$) in Case 3. Random selection was also significantly slower than the proposed method, requiring $5.62$, $5.66$, and $5.81$ rounds on average in Cases~1--3.

\begin{table}[t]
\centering
\begin{threeparttable}
\caption{Convergence comparison across methods and 50 seeds}
\label{tab:results-table}
\begin{tabular}{llccc}
\toprule
Case & Method & Avg. Rounds\tnote{*} & Success \\
\midrule
\multirow{3}{*}{Case 1} 
 & Our Method        & 3.14 & 50/50 &\\
 & Model Uncertainty & 3.65 & 49/50 &\\
 & Random            & 5.62 & 48/50 \\
\midrule
\multirow{3}{*}{Case 2} 
 & Our Method        & 2.86 & 49/50 \\
 & Model Uncertainty & 4.24 & 49/50 \\
 & Random            & 5.66 & 47/50 \\
\midrule
\multirow{3}{*}{Case 3} 
 & Our Method        & 3.59 & 49/50 \\
 & Model Uncertainty & 4.50 & 50/50 \\
 & Random            & 5.81 & 48/50 \\
\bottomrule
\end{tabular}
\begin{tablenotes}
\footnotesize
\item[$^{*}$] Average only including successful runs.
\end{tablenotes}

\end{threeparttable}
\end{table}

\section{Discussion}

In this paper, we propose a framework that shifts experimental design from model identification to controller identification. When multiple hypotheses induce the same optimal controller, reducing model uncertainty may not affect the controller choice. By minimizing posterior expected regret, our approach prioritizes experiments that resolve controller ambiguity rather than global model uncertainty, and stops once the remaining uncertainty has only a limited impact on controller performance.

Several limitations should be acknowledged. The method optimizes only over the provided controller library, and therefore performance is limited by the best controller in $\mathcal{C}$, i.e. $\max_{c\in\mathcal C}J(h^\star,g_c)$. In addition, the finite-stopping guarantee requires $\delta_\epsilon>0$, so arbitrary choices of $\mathcal{E}$ may not guarantee stopping. The theoretical formulation also assumes that the hypothesis set contains the plant model, $h^\star\in\mathcal H$, which may not always hold, although controller selection may still be effective when candidate hypotheses capture the dynamics relevant to the controller choice.

As future work, extending the framework to adaptively refine the hypotheses and controller spaces remains an important direction. Experiment selection can also become computationally expensive as the hypothesis, controller, or experiment spaces grow, which is observed in related Bayesian design methods~\cite{zhongGoalOrientedBayesianOptimal2026, fosterUnifiedStochasticGradient2020}, so hypothesis pruning or reduced sampling may be needed at larger scales. Finally, the current one-step lookahead may miss informative multi-experiment strategies, motivating extensions to longer planning horizons.

\section{Acknowledgements}

This work was partially supported by the NSF under grants GCR 2219101, EFRI 2422282 
and by a Brendan Iribe Endowed Professorship.

\vfill 
\bibliographystyle{IEEEtran}
\bibliography{references}

\end{document}